\documentclass[12pt]{article}
\usepackage{latexsym,amsfonts,amssymb,amsmath,amsthm}
\usepackage{graphicx}

\usepackage[unicode,bookmarksnumbered=true]{hyperref}
\newtheorem{theorem}{Theorem}

\newcommand{\be} {\begin{equation}}
\newcommand{\ee} {\end{equation}}
\newcommand{\ba} {\begin{array}{l}}
\newcommand{\ea} {\end{array}}
\newcommand{\p} {\partial}
\newcommand{\lbd} {\lambda}
\newcommand{\al} {\alpha}

\begin{document}
\begin{center}\LARGE
Some remarks about  Lie and  potential symmetries of a class of Korteweg-de Vries  type  equations
\end{center}

Oleksii Pliukhin$^1$, Danny Arrigo$^2$ and Roman Cherniha$^3$
\vspace{.5cm}

$^1$ Poltava National Technical Yuri Kondratyuk University,

24, Pershotravnevyj Ave., Poltava 36011, Ukraine
\vspace{.2cm}

$^2$ University of Central Arkansas

201 Donaghey Ave., Conway, AR 72035, USA
\vspace{.2cm}

$^3$ Institute of Mathematics, Ukrainian National Academy of Sciences

3, Tereshchenkivs'ka Street, Kyiv 01004, Ukraine

\begin{abstract}
Preliminary  results about Lie  and potential  symmetries of a class of Korteweg-de Vries type equations  are presented. In order to prove existence of potential symmetries three different systems of so called determining equations are analysed. It is  shown that  two systems lead only to Lie symmetries while the third system produces potential symmetries provided  the equation in question has an appropriate structure.
\end{abstract}

\section{Introduction.}

The main object of this paper is the class of  the third-order  nonlinear PDEs
\be\label{q0}u_t=[A(u)u_{xx}+B(u)u_x+C(u)]_x,\ee
where $u=u(t,x)$  is an unknown function and $A(u), B(u)$  and $C(u)$  are arbitrary smooth functions
(hereafter  the lower subscripts $t$  and $x$ denote differentiation  with respect to these variables).
Obviously, \eqref{q0}  with $A(u)=0$  reduces to the class of reaction-convection equations, which was studied for long time by different mathematical techniques including symmetry-based methods (see  the recent book \cite{cher-sero-pliu-2018} and the references cited therein). We assume  $A(u)\not =0$ in what follows and do not consider the cases when \eqref{q0} is a linear PDE.

The PDE class \eqref{q0} contains as particular cases several well-known equations arising in applications.
Setting $A(u)=1, B(u)=0$  and  $C(u)=u^2$ one obtains the celebrated Korteweg-de Vries (KdV) equation
\be\label{0-1}u_t= u_{xxx} + 2uu_x,\ee
which is the classical example of integrable equation \cite{gard-gree-krus-miur-1967} and was investigated in a huge number of papers and several books is devoted to this equations (see, e.g., \cite{ablo-segu-1981}, \cite{das-1989}, \cite{miur-1976} and the papers cited therein). In the case  $A(u)=1, B(u)=0$  and  $C(u)=u^3$ one obtains the modified  Korteweg-de Vries (mKdV), which is related with the KdV  equation.

In the case $A(u)=1, B(u)=0$  and  $C(u)=u^m$,  the well-known $K(m,1)$ equation is obtained, which is a natural generalization of the KdV  equation.  Setting  $A(u)=-u^n, B(u)=0$  and  $C(u)=-u^m$  one derives an analog of  the well-known $K(m,n)$ equation
\be\label{0-2}u_t + (u^n)_{xxx} + (u^n)_x =0,\ee
which was introduced in \cite{rose-hyma-1993} as an example of  the KdV type  equation possessing so-called compactons (solitons with a compact support).  Eq.\eqref{0-2} was studied in several papers with aim to construct exact solutions \cite{brac-2004}, \cite{brac-2005}, \cite{rose-hyma-1993}, \cite{rose-zilb-2018}.

We also point out that the KdV-Burgers equation
\be\label{0-3}u_t= u_{xxx} + \lambda u_{xx} +\mu  uu_x\ee
belongs  to the  PDE class \eqref{q0}. This equation describes propagation of undular bores in shallow water \cite{benn-1966}, \cite{john-1972} and was studied in a number of papers (see, e.g. \cite{chen-2000} and references cited therein).

In the case $A(u)=1, B(u)=\lambda $  and  $C(u)=\mu u^m$, one obtains the generalized KdV-Burgers equation \cite{chen-2000}, \cite{pego-smer-wein-1993}.

 Investigation of  Lie symmetry of the KdV equation and its generalizations  was initiated many years ago.
 To the best of our knowledge, paper  \cite{ga-wi-92} is one of pioneering  papers in this direction (see also recent papers \cite{char-vane-soph-2014},
\cite{vane-post-2017}). Notably, all these papers are devoted to the cases, when the KdV type  equations contain variable coefficients depending on the time or/and space variables.
   However, the Lie symmetry classification of the PDE class \eqref{q0} was not yet under study to the best of our knowledge. Of course, Lie symmetry of  some particular equations belonging to this class is well-known (including KdV and $K(m,1)$ equations). Moreover, we are interested also in potential symmetries of equations belonging to the class \eqref{q0}.

Because a complete classification of Lie and potential symmetries of the  PDE class \eqref{q0} is a highly nontrivial problem, here we report  some particular results. Inspired by the recent paper~\cite{arri-ashl-bloo-16},
we started from comparing the determining equations for search Lie and potential symmetries.
After rigorous  analysis,  we derived several statements, which are proved in Section~2. In Section 3, conclusions and a plan of following investigations are presented.

\section{Main results.}

We are interested in comparing of Lie symmetries and potential  symmetries of equations  belonging to the PDE class ~\eqref{q0}. It is well-known that both types of symmetries can be derived by application of the classical Lie scheme,
which are described in many books and textbooks (see, for example, \cite{cher-sero-pliu-2018}, \cite{fu-sh-se-93}, \cite{ol-93}, \cite{ov-82}).
First of all it should be mentioned that the notion of potential symmetry is closely related to that of Lie symmetry.
In particular, each Lie symmetry of a given PDE is automatically a potential symmetry, however the vice versa statement, generally speaking,  is wrong. The notion of  the  potential symmetry was introduced in 1980s \cite{bl-re-ku-88}, \cite{blum-reid-kume-1988(E)}, \cite{kr-vi-89} and several nontrivial examples were simultaneously discovered.

We start from the most general form of the Lie symmetry operator
\be\label{q00} X=\xi^0(t,x,u)\p_t+\xi^1(t,x,u)\p_x+\eta(t,x,u) \p_{u}\ee
 of a given  equation of the form \eqref{q0}, in which the coefficient $\xi^0(t,x,u),  \xi^1(t,x,u)$ and
 $\eta(t,x,u)$ are some smooth functions.

 In order to find  potential symmetries of the given equation, one should introduce the new depended variable $v(t,x)$ using the  non-local  substitution $u=v_x$ and then the given PDE of the form \eqref{q0} is transformed into
 the system
 \be\label{q3} v_t=A(u)u_{xx}+B(u)u_x+C(u),\ v_x=u.\ee

Obviously, the most  general form of the Lie symmetry operator of this system possesses the form
\be\label{q2} X=\xi^0(t,x,u,v)\p_t+\xi^1(t,x,u,v)\p_x+\eta^1(t,x,u,v) \p_{u}+\eta^2(t,x,u,v)\p_{v}.\ee

On the other hand, any operator \eqref{q2} of Lie's invariance of system \eqref{q3},
 which contains at least one coefficient among $\xi^0(t,x,u,v),  \xi^1(t,x,u,v)$ and
 $\eta(t,x,u,v)$, which really depends on $v$, is a pure potential symmetry for the given scalar equation.


Now we apply the classical Lie algorithm for construction so-called system of determining equations (DEs) of the given PDE of the form \eqref{q0} (see, for example, \cite{cher-sero-pliu-2018}, \cite{fu-sh-se-93}, \cite{ol-93}, \cite{ov-82}).
This system  has the form
\begin{subequations}\begin{align}\label{q31}
&(\al_x+\xi^1_{xx})A'=0,\\\label{q32}
&(\al u+\beta)A'=(3\xi^1_x-\xi^0_t)A,\\\label{q33}
&(\al u+\beta)B'=(2\xi^1_x-\xi^0_t)B-3\al_x A,\\\label{q34}
&(\al_x+\xi^1_{xx})C'=(\al_{xx}+\xi^1_{xxx})B -(\al_{xxx}+\xi^1_{xxxx})A+\al_t+\xi^1_{xt},\\\label{q35}
&(\al_{x}u+\beta_{x})C'=-(\al_{xx}u+\beta_{xx})B- (\al_{xxx}u+\beta_{xxx})A+\al_tu+\beta_t,\end{align}
\end{subequations}
where $\xi^0=\xi^0(t),\ \xi^1=\xi^1(t,x),\  \eta=\al(t,x) u+\beta,\ \beta=\beta(t,x).$

System of DEs for finding  Lie symmetries of system~\eqref{q3}
 has the form
\begin{subequations}\begin{align}\label{q41}
&(\al u+\gamma_x)A'=(3\xi^1_x-\xi^0_t)A, \\ \label{q42}
&(\al u+\gamma_x)B'=(2\xi^1_x-\xi^0_t)B-3\al_x A,\\ \label{q43}
&(\al u+\gamma_x)C'=(C_2-\xi^0_t)C-(\al_x u+\gamma_{xx})B- (\al_{xx}u+\gamma_{xxx})A-\xi^1_t u+\eta^2_t,\\
\label{q5}&\al=C_2-\xi^1_x,\end{align}
\end{subequations}
where $\xi^0=\xi^0(t),\ \xi^1=\xi^1(t,x),\  \eta^1=\al(t,x) u+\gamma_x,\ \gamma=\gamma(t,x),\ \eta^2=C_2 v+\gamma,\ C_2=const$.

System~\eqref{q41} was derived by straightforward application of the Lie algorithm. However, system~\eqref{q3} consists of PDEs of different order and one should take
into account  differential consequences of PDE of lower order   in order to identify  all possible Lie symmetries. So, take
into account  the differential consequence with respect to $x$,
 we obtain the three-component  system
\be\label{q4}v_t=A(u)u_{xx}+B(u)u_x+C(u),\ v_x=u,\ v_{xx} = u_x.\ee

\textbf{Remark.} We do not use the second differential consequence  $v_{xt} = u_t$  because it can be  shown that system of  DEs is still the same as for system~\eqref{q4}.


System of DEs for finding  Lie symmetries of system~\eqref{q4}
(i.e.  system~\eqref{q3} with the differential consequence of the second equation) has the form
\begin{subequations}\begin{align}\label{q81}
&(\al u+\beta)A'=(3\xi^1_x-\xi^0_t)A, \\ \label{q82}
&(\al u+\beta)B'=(2\xi^1_x-\xi^0_t)B-3(\al_v u+\al_x)A,\\ \label{q83}
&(\al u+\beta)C'=(\al+\xi^1_x-\xi^0_t)C-[\al_v u^2+(2\al_x +\xi^1_{xx})u+\beta_x]B- \\&\nonumber-[\al_{vv} u^3 +3\al_{vx}u^2+(3\al_{xx}+2\xi^1_{xxx})u+\beta_{xx}]A-\xi^1_t u+\eta^2_t,\\ \label{q84} &\eta^2_v=\al+\xi^1_x,\ \eta^2_x=\beta,
\end{align}
\end{subequations}
where $\xi^0=\xi^0(t),\ \xi^1=\xi^1(t,x),\  \eta^1=\al(t,x,v) u+\beta,\ \beta=\beta(t,x,v), \ \eta^2=\eta^2(t,x,v).$

\begin{theorem}\label{th1}
The system of DEs \eqref{q31}--\eqref{q35}  for finding Lie symmetries of Eq.~\eqref{q0} is equivalent to the system of DEs  \eqref{q41}--\eqref{q5} for finding Lie symmetries of  system  \eqref{q3}.
\end{theorem}

\begin{proof}
Let us compare systems \eqref{q31}--\eqref{q35} and \eqref{q41}--\eqref{q5}.

First, we will show, that system~\eqref{q31}--\eqref{q35} is a consequence of system~\eqref{q41}--\eqref{q5}.

Equations \eqref{q31} and \eqref{q34} vanish with \eqref{q5}.

Equations \eqref{q32}, \eqref{q41} and Eqs. \eqref{q33}, \eqref{q42} coincide providing $\beta=\gamma_x$.

Finally, differentiating \eqref{q43} with respect to $x$, we arrive exactly at Eq.~\eqref{q35}.

So, system of determining equations for finding potential symmetries \eqref{q41}--\eqref{q5} of Eq.~\eqref{q0} contains in system of determining equations for finding Lie symmetries \eqref{q31}--\eqref{q35} of Eq.~\eqref{q0}.

To show  that system~\eqref{q41}--\eqref{q5} is a consequence of system~\eqref{q31}--\eqref{q35} one needs to discuss two cases, namely:

1) $\al_x=-\xi^1_{xx}$,

2) $A=\lbd=const$.

In case 1) Eq.~\eqref{q32}, \eqref{q41} and Eqs. \eqref{q33}, \eqref{q42} coincide providing $\beta=\gamma_x$ and $\al_x=-\xi^1_{xx}$. Integrating \eqref{q35} with respect to $x$, taking into account \eqref{q34} and $\beta=\gamma_x$, we obtain
\be\label{q1}(\al u+\gamma_x)C'=F(t,u)-(\al_x u+\gamma_{xx})B-(\al_{xx} u+\gamma_{xxx})A- \xi^1_t u+\gamma_t,\ee
where $F(t,u)$ -- arbitrary function.
Setting $F(t,u)=(C_2-\xi^0_t)C$ we obtain exactly Eq.~\eqref{q43}.

Case 2). It turns out that we need to prove that $\al_x=0$ in order to show that system  \eqref{q41}--\eqref{q5} is a consequence of system~\eqref{q31}--\eqref{q35}. In the case $\al_x = 0,\ A = \lbd = const$ system \eqref{q31}--\eqref{q35} takes the form
\begin{subequations}\begin{align}
\label{q6}&\xi^1_x=\frac 13 \xi^0_t,\\
\label{q48}
&(\al u+\beta)B'=-\frac13\xi^0_tB,\\\label{q49}
&\al_t = -\xi^1_{tx},\\\label{q50}
&\beta_{x}C'=-\beta_{xx}B- \lbd \beta_{xxx} +\al_t u+\beta_t.\end{align}
\end{subequations}

Equation \eqref{q41} vanishes, providing \eqref{q6}; Eq. \eqref{q42} coincides with \eqref{q48} providing $\beta = \gamma_x$; integrating \eqref{q49} with respect to $t$, we obtain \eqref{q5}; integrating \eqref{q50} with respect to $x$, and taking into account $\beta = \gamma_x$ and \eqref{q49}, we obtain \be\label{q12}\beta C'=F(t,u)-\beta_x B-\lbd \beta_{xx}-\xi^1_t u+\gamma_t.\ee

Setting $F(t,u)=-\al u C'+(C_2-\xi^0_t)C$ we obtain exactly Eq.~\eqref{q43}.

Let us show that  $\al_x = 0$. Setting $A=const$ in Eq.~\eqref{q32}, we immediately obtain \eqref{q6}.

Using \eqref{q6}, we rewrite Eqs. \eqref{q33}--\eqref{q35} in the following form
\begin{subequations}\begin{align}
\label{q61}
&(\al u+\beta)B'=-\frac13\xi^0_tB-3\lbd\al_x,\\\label{q62}
&\al_xC'=\al_{xx}B -\lbd\al_{xxx}+\al_t+\xi^1_{tx},\\\label{q63}
&(\al_{x}u+\beta_{x})C'=-(\al_{xx}u+\beta_{xx})B- \lbd (\al_{xxx}u+\beta_{xxx})+\al_t u+\beta_t.\end{align}
\end{subequations}

Differentiating \eqref{q61} w.r.t. $x$, we arrive at
\be\label{q7}(\al_x u+\beta_x)B'=-3\lbd \al_{xx}.\ee
In the case $B'\ne 0, \frac{\lbd_1} u$, we obtain $\al_x=0$. So, we need to discuss $B'=\frac {\lbd_1 }u\ne0$ and $B' = 0$.

$B'=\frac {\lbd_1 }u\ne0$ gives us
\be\label{q8}B=\lbd_1 \ln u+\lbd_2.\ee

Substituting \eqref{q8} into \eqref{q61}, we obtain
\be\label{q9}\lbd_1 \al+3\lbd \al_x+\lbd_1 \beta \frac 1 u+\frac13 \lbd_1 \xi^0_t \ln u+\frac13\lbd_2 \xi^0_t=0.\ee
Splitting Eq.~\eqref{q9} w.r.t. different functions of $u$, we obtain
\be\label{q10}\xi^0_t=\beta=0,\ \al=f(t) e^{-\frac {\lbd_1}{3\lbd}x}.\ee
Substituting \eqref{q10} into \eqref{q62}--\eqref{q63}, we derive
\begin{subequations}\begin{align}\label{q71}
&9 \lambda  \lbd_1 f(t) C'(u)+\lbd_1^3 f(t) (1+3\ln u)+27 \lambda ^2 f'(t)+3 \lbd_1^2 \lbd_2 f(t)=0,\\\label{q72}&9 \lambda  \lbd_1 f(t) C'(u)+\lbd_1^3 f(t) (1-3\ln u)+27 \lambda ^2 f'(t)-3 \lbd_1^2 \lbd_2 f(t)=0.
\end{align}
\end{subequations}
Subtracting \eqref{q71} and \eqref{q72}, we obtain
\[6 \lbd_1^2 f(t) (\lbd_1 \ln u+\lbd_2)=0.\]

Obviously, there is just one possibility $f(t)=0$, which leads to $\al=0$.

So, we proved that $\al_x=0$.





Finally, we should discuss $B'=0$. All subcases, except one, give us $\al_x=0$. In the case $\al_x\ne0$ we obtain the linear equation
\[u_t = \lbd u_{xx} + \mu u_{xx} + \gamma_1 u + \gamma_2,\]
which contradicts to our assumption about examination of nonlinear PDEs only.

The proof is now completed.
\end{proof}


\textbf{Consequence 1.} Lie symmetries of any nonlinear PDE system of the form  \eqref{q3} cannot generate a potential symmetry for the relevant scalar equation of the form  \eqref{q0}.

\begin{theorem}\label{th2}
The system of determining equations  \eqref{q31}--\eqref{q35}  for finding Lie symmetries of Eq.~\eqref{q0}  with   $A \ne const$  is equivalent
to the system  \eqref{q81}--\eqref{q84} for finding Lie symmetries of  system  \eqref{q4} with $A\ne const.$

\end{theorem}

\begin{proof}
Let us compare systems \eqref{q31}--\eqref{q35} and \eqref{q81}--\eqref{q84}. In quite similar way, as for the proof of Theorem~\ref{th1}, we may prove that system~\eqref{q81}--\eqref{q84} is a consequence of system~\eqref{q31}--\eqref{q35}.

Finally, we should prove that system~\eqref{q31}--\eqref{q35} is a consequence of system~\eqref{q81}--\eqref{q84}.

Differentiating Eq.~\eqref{q81} with respect to $v$, we derive
\be\label{q37}\al_v=0,\ \beta_v=0.\ee

Integrating \eqref{q84} and taking into account \eqref{q37}, we obtain
\be\label{q38}\eta^2=(\al+\xi^1_x)v+f(t,x).\ee

Differentiating \eqref{q38} with respect to $x$ and taking into account \eqref{q84}, we arrive at
\be\label{q39}(\al_x+\xi^1_{xx})v+f_x=\beta.\ee

Differentiating \eqref{q39} with respect to $v$ and taking into account \eqref{q37}, we derive
\be\label{q40}\al_x=-\xi^1_{xx},\ \beta = f_x(t,x).\ee
So, Eq.~\eqref{q31} vanishes providing \eqref{q40}; Eq.~\eqref{q32} coincides with Eq.~\eqref{q81}; Eq.~\eqref{q33} coincides with Eq.~\eqref{q82} providing \eqref{q37}; Eq.~\eqref{q34} gives us $\al_t+\xi^1_{tx}=0$ the same condition one can obtain from \eqref{q83} differentiating it with respect to $v$ and taking into account \eqref{q37}, \eqref{q38} and \eqref{q40}; Eq.~\eqref{q35} we may obtain differentiating Eq.~\eqref{q83} with respect to $x$ and taking into account \eqref{q37} and \eqref{q40}.

The proof is now completed.
\end{proof}

\textbf{Consequence 2.} Lie symmetries of any nonlinear PDE system of the form  \eqref{q4} with $A\ne const$ cannot generate a potential symmetry for the relevant scalar equation of the form  \eqref{q0} with $A\ne const$.

\begin{theorem}\label{th3}
The system of determining equations \eqref{q31}--\eqref{q35} for finding Lie symmetries of Eq.~\eqref{q0}  with $A = const$  is NOT equivalent
to the system~\eqref{q81}--\eqref{q84} for finding Lie symmetries of  system  \eqref{q4} with   $A = const$.
Moreover,   system  \eqref{q81}--\eqref{q84}  produces  potential  symmetries of Eq.~\eqref{q0} with correctly-specified functions $B$ and $C$.

\end{theorem}

\begin{proof}
In order to prove this theorem, it is enough to present an example of a Lie symmetry of \eqref{q4} (with correctly-specified functions $B$ and $C$) leading to a pure potential symmetry of the relevant equation of the form \eqref{q0}.
So,
 we are going to present a particular solution of system~\eqref{q81}--\eqref{q84} with correctly-specified functions $B$ and $C$, which is not obtainable from system~\eqref{q31}--\eqref{q35}.


We have noted that the particular solution of system~\eqref{q81}--\eqref{q84} with
\be\label{q45}A = \lbd,\ B= \mu u,\ C = \frac{\mu^2}{9 \lbd} u^3\ee
has the form
\be\label{q44}\xi^0 = C_0,\ \xi^1 = C_1,\ \eta^1 = \mu C_2 e^{-\frac \mu {3\lbd}v} u,\ \eta^2 = -3 \lbd C_2 e^{-\frac \mu {3\lbd}v} + C_3.\ee

Substituting \eqref{q45} and \eqref{q44} into \eqref{q4} and \eqref{q2} respectively, we obtain the system
\be\label{q45a} v_t = \lbd u_{xx} + \mu u u_x + \frac{\mu^2}{9 \lbd} u^3,\ v_x = u, \ v_{xx} = u_x.\ee
which is invariant under the operator
\be\label{q46}X = C_0 \p_t + C_1 \p_x + \mu C_2 e^{-\frac \mu {3\lbd}v} u \p_u + \Bigr (-3 \lbd C_2 e^{-\frac \mu {3\lbd}v} + C_3\Bigr)\p_v.\ee

So, \eqref{q46} is the Lie symmetry of  system \eqref{q45a}.
On the other hand, the coefficient in front of $\p_u$ in \eqref{q46} explicitly depends on $v$, therefore operator \eqref{q46} is not obtainable from system~\eqref{q31}--\eqref{q35}. So, operator \eqref{q46} is the operator of pure potential  symmetry of the equation
\be\label{q47} u_t = (\lbd u_{xx} + \mu u u_x + \frac{\mu^2}{9 \lbd} u^3)_x.\ee
More precisely, operator  \be\label{q46a}X = C_0 \p_t + C_1 \p_x + \mu C_2 e^{-\frac \mu {3\lbd}v} u \p_u \ee
with $v_x = u$  is a potential symmetry of Eq.~\eqref{q47}.

The proof is now completed.
\end{proof}

It is interesting to note that Eq.~\eqref{q47} with $\mu=0$ is nothing else but the modified KdV equation. However, the potential symmetry \eqref{q46a} with $\mu=0$ degenerates to the Lie symmetry
\[X = C_0 \p_t + C_1 \p_x, \]
which corresponds to the time and space translations.

\section{Conclusions.}

In this paper,  it was shown that potential symmetries of  equations belonging to the PDE class~\eqref{q0} cannot be identified in the standard way  using the relevant system \eqref{q3}. In fact, we have proved
that the system of DEs \eqref{q41}--\eqref{q5} for finding Lie  symmetries of  system \eqref{q3}
 is equivalent to the system of DEs \eqref{q31}--\eqref{q35} for finding Lie symmetries of Eq.~\eqref{q0}. This result is presented in Theorem \ref{th1}.

The situation is different if one takes into account the fact that system \eqref{q3}  consists of PDEs of {\it different order}, therefore  one should take
into account  differential consequences of PDE of lower order.
As a result it was proved the similar result (see Theorem \ref{th2}), however under the restriction $A \ne const$.

In contradiction to the case $A \ne const$,  the case $A = const$ is quite different. We have proved  that systems of DEs \eqref{q81}--\eqref{q84} and \eqref{q31}--\eqref{q35} are inequivalent (see, Theorem \ref{th3}).
 Moreover, a highly nontrivial  example of potential Lie symmetry of Eq.~\eqref{q0}, namely Eq.~\eqref{q47},  possessing  the potential symmetry  \eqref{q46} was found.

 It should be stressed that the issue discovered above does not occurs in the case of reaction-convection equations because the relevant system for finding potential symmetries consists of PDEs of {\it the same  order}. In particular, it occurs in the case of nonlinear diffusion equation \cite{bl-re-ku-88}.

The work is in progress. In a forthcoming paper,  we are going to present  a complete Lie  classification   and  a complete potential  symmetry classification of Eq.~\eqref{q0} together with some new  exact solutions.


\end{document}